\documentclass[11pt]{article}
\usepackage{authblk}
\usepackage{booktabs}
\usepackage[most]{tcolorbox}

\title{Maxwell's equations revisited - mental\\
imagery and mathematical symbols}
\author[1]{Matthias Geyer}
\author[2]{Jan Hausmann}
\author[2]{Konrad Kitzing}
\author[2]{\\ Madlyn Senkyr}
\author[2]{Stefan Siegmund}
\affil[1]{Institute for Materials Science, TU Dresden, Germany}
\affil[2]{Faculty of Mathematics, TU Dresden, Germany}

\usepackage{cite}

\usepackage{amsmath,amssymb,amsthm,mathtools,bbm}
\usepackage[most]{tcolorbox}
\usepackage{booktabs}
\providecommand{\definitionname}{Definition}
\providecommand{\theoremname}{Theorem}
\providecommand{\propositionname}{Proposition}
\providecommand{\lemmaname}{Lemma}
\providecommand{\corollaryname}{Corollary}
\providecommand{\remarkname}{Remark}
\providecommand{\examplename}{Example}
\providecommand{\historicalremarkname}{Historical Remark}

\theoremstyle{definition}
\newtheorem{definition}{\protect\definitionname}
\theoremstyle{plain}
\newtheorem{theorem}[definition]{\protect\theoremname}

\newtheorem{lemma}[definition]{\protect\lemmaname}

\theoremstyle{remark}
\newtheorem{remark}[definition]{\protect\remarkname}
\theoremstyle{definition}

\theoremstyle{definition}
\newtheorem{historicalremark}{\protect\historicalremarkname}
\numberwithin{equation}{section}
\allowdisplaybreaks
\usepackage{enumitem}
\setlist[enumerate]{label*=(\alph*),ref=(\alph*)}

\newcommand{\R}{\mathbb{R}}

\let\det\undefined
\DeclareMathOperator{\det}{det}
\DeclareMathOperator{\dom}{dom}

\renewcommand{\phi}{\varphi}

\begin{document}

\maketitle

\begin{abstract}
Using Maxwell's mental imagery of a tube of fluid motion of an imaginary fluid, we derive his equations 
$\operatorname{curl} \mathbf{E} =  -\frac{\partial \mathbf{B}}{\partial t}$,
$\operatorname{curl} \mathbf{H} = \frac{\partial \mathbf{D}}{\partial t} + \mathbf{j}$,
$\operatorname{div} \mathbf{D} = \varrho$,
$\operatorname{div} \mathbf{B} = 0$, which together with the constituting relations $\mathbf{D} = \varepsilon_0 \mathbf{E}$, $\mathbf{B} = \mu_0 \mathbf{H}$, form what we call today Maxwell's equations. Main tools are the divergence, curl and gradient integration theorems and a version of Poincare's lemma formulated in vector calculus notation. Remarks on the history of the development of electrodynamic theory, quotations and references to original and secondary literature complement the paper.
\end{abstract}

\emph{Keywords and phrases.} Maxwell's equations

\emph{Mathematics subject classification.} 78A25 

\maketitle


\section{Introduction} \label{S:label0}

\noindent
Maxwell developed his famous equations
\begin{gather*}
   \operatorname{curl} \mathbf{E} + \frac{\partial \mathbf{B}}{\partial t} = 0,
   \quad
   \operatorname{div} \mathbf{D} = \varrho,
\\
   \operatorname{curl} \mathbf{H} - \frac{\partial \mathbf{D}}{\partial t} = \mathbf{j},
   \quad
   \operatorname{div} \mathbf{B} = 0,
\\[0.8ex]
   \mathbf{D} = \varepsilon_0 \mathbf{E},
   \quad
   \mathbf{B} = \mu_0 \mathbf{H},
\end{gather*}
in what Hon and Goldstein \cite{HonGoldstein2020} call an odyssey in electromagnetics consisting of four stations: 
\pagebreak
\begin{itemize}
   \item Station 1 (1856-1858): on Faraday's lines of force \cite{Maxwell1858}
   \item Station 2 (1861-62): on physical lines of force \cite{Maxwell1861-62}
   \item Station 3 (1865): A dynamical theory of the electromagnetic field \cite{Maxwell1865}
   \item Station 4 (1873): A treatise on electricity and magnetism \cite{Maxwell1873}
\end{itemize}
Maxwell's original work is a rich source for methodological inspiration \cite{MaxwellNiven1890, Maxwell1873}. 
Many excellent books \cite{Darrigol2003, HonGoldstein2020, Siegel1991} and papers \cite{Arthur2013,Bork1963,Turnbull2013} have been written on Maxwell's construction of an electrodynamic theory, the transcription of Maxwell's equations to vector analysis notation \cite{Hunt2012, Unz1963} and their contemporary presentation with differential forms \cite{Schleifer1983,WarnickEtAl1997,Lindel2004,FumeronEtAl2020}. When it comes to teaching Maxwell's equations many aspects of the didactic transposition process from Maxwell's original work, to textbooks and then to classroom are lost. As pointed out in \cite{KaramEtal2014} this is, in particular, the case for the displacement current term $\frac{\partial \mathbf{D}}{\partial t}$ insertion. In this article we present a method for introducing Maxwell's displacement current concept, within the framework of an integral formulation of Maxwell's equations (for other didactic derivations of Maxwell's equations, see e.g.\ \cite{DienerEtal2013,Gauthier1983, KaramEtal2014}). Our method is deliberately not formulated in the elegant and more advanced differential form calculus, but the readily available vector calculus notation. We stay very close to Maxwell's original work in the following specific aspects:
\begin{itemize}
   \item We start out with Maxwell's first paper on Faraday's lines of force \cite{Maxwell1858}. 
   \item We apply his analogy of an imaginary fluid to electricity and magnetism.
   \item We use his mental imagery of a tube of fluid motion.
\end{itemize}
Historical Remarks are added for the reader's convenience and as references to Maxwell's work and secondary literature. Mathematical preliminaries are the well-known integration theorems
\begin{equation*}\textstyle
    f_1 - f_2
    =
    \!
    \int\limits_{\mathrm{curve}}
    \operatorname{grad}
    f \cdot
    d \mathbf{s}
    ,
    \int\limits_{\mathrm{curve}}
    \mathbf{a} \cdot
    d \mathbf{s}
    =
    \!
    \int\limits_{\mathrm{surface}}
    \operatorname{curl}
    \mathbf{a} \cdot
    d \mathbf{S}
    ,
    \int\limits_{\mathrm{surface}}
    \mathbf{b} \cdot
    d \mathbf{S}
    =
    \!
    \int\limits_{\mathrm{volume}}
    \operatorname{div}
    \mathbf{b}
    \, d V
\end{equation*}
which are recalled in Theorem \ref{thm:integration}, together with a version of Poincare`s Lemma.

\pagebreak
\begin{historicalremark}[Mental Imagery \& Mathematical Symbols]\label{rem:MentalImagery}
On September 15th, 1870, Maxwell  gave an \emph{Address to the Mathematical and Physical Sections of the British Association} (see e.g.\ the reprint in \cite[vol.\ 2, pp.\ 218-219]{MaxwellNiven1890}) in which he also explains the role of mental imagery and mathematical symbols:

\begin{quote}\footnotesize
The human mind is seldom satisfied, and is certainly never exercising its highest functions, when it is doing the work of a calculating machine. What the man of science, whether he is a mathematician or a physical inquirer, aims at is, to acquire and develop clear ideas of the things he deals with. $[\dots]$

But if he finds that clear ideas are not to be obtained by means of processes the steps of which he is sure to forget before he has reached the conclusion, it is much better that he should turn to another method, and try to understand the subject by means of well-chosen illustrations derived from subjects with which he is more familiar.

$[\dots]$ [A] truly scientific illustration is a method to enable the mind to grasp some conception or law in one branch of science, by placing before it a conception or a law in a different branch of science, and directing the mind to lay hold of that mathematical form which is common to the corresponding ideas in the two sciences $[\dots]$.

The correctness of such an illustration depends on whether the two systems of ideas which are compared together are really analogous in form, or whether, in other words, the corresponding physical quantities really belong to the same mathematical class. When this condition is fulfilled, the illustration is not only convenient for teaching science in a pleasant and easy manner, but the recognition of the formal analogy between the two systems of ideas leads to a knowledge of both, more profound than could be obtained by studying each system separately.

$[\dots]$ [S]cientific truth should be presented in different forms, and should be regarded as equally scientific, whether it appears in the robust form and the vivid colouring of a physical illustration, or in the tenuity and paleness of a symbolical expression.

Time would fail me if I were to attempt to illustrate by examples the scientific value of the classification of quantities. I shall only mention the name of that important class of magnitudes having direction in space which Hamilton has called vectors, and which form the subject-matter of the Calculus of Quaternions, a branch of mathematics which, when it shall have been thoroughly understood by men of the illustrative type, and clothed by them with physical imagery, will become, perhaps under some new name, a most powerful method of communicating truly scientific knowledge to persons apparently devoid of the calculating spirit.
\end{quote}
\vspace*{-6ex}\hfill $\diamond$
\end{historicalremark}

In this paper we adopt Maxwell's empathetic perception that scientific truth should be presented in different forms, using mental imagery as well as mathe\-matical symbols. We follow his analogy between a model of an imaginary fluid and his equations of electrodynamics which are indeed analogous in form and belong to the same mathematical class. A truly scientific illustration taking the role of mental imagery will suggest itself by visualizing the concepts corresponding to the mathematical symbols used to describe the analogy, in our case this will be the divergence, curl and gradient integration theorems and a geometric interpretation of Poincar\'{e}'s Lemma in terms of the integration theorems.

\begin{historicalremark}[Quaternions \& Vector Analysis]
Maxwell's prevision that the Calculus of Quaternions under some new name would become a most powerful method of communicating scientific knowledge (see Remark \ref{rem:MentalImagery}) came true with the early development of vector analysis \cite{Crowe1994} from quaternions 
\begin{equation*}
   q_0
   +
   q_1 \mathbf{i}
   +
   q_2 \mathbf{j}
   +
   q_3 \mathbf{k}
   \quad
   (q_0, q_1, q_2, q_3 \in \R)
   \qquad
   \text{with }
   \mathbf{i}^2 = \mathbf{j}^2 = \mathbf{k}^2 = \mathbf{ijk} = -1.
\end{equation*}
Hamilton who discovered the quaternions in 1843 noted already the special role of what he called a \emph{vector} 
\begin{equation*}
   q_1 \mathbf{i}
   +
   q_2 \mathbf{j}
   +
   q_3 \mathbf{k}
\end{equation*}
compared to what he called the \emph{scalar part} $q_0$ of $q \coloneqq q_0 + q_1 \mathbf{i} + q_2 \mathbf{j} + q_3 \mathbf{k}$. His notation $Sq = q_0$ and $Vq = q_1 \mathbf{i} + q_2 \mathbf{j} + q_3 \mathbf{k}$ applied to the product of two vectors
\begin{gather*}
   S(x \mathbf{i} + y \mathbf{j} + z \mathbf{k})(x' \mathbf{i} + y' \mathbf{j} + z' \mathbf{k})
   =
   -(xx' + yy' + zz'),
\\
   V(x \mathbf{i} + y \mathbf{j} + z \mathbf{k})(x' \mathbf{i} + y' \mathbf{j} + z' \mathbf{k})
   =
   (y z' - z y') \mathbf{i} + (z x' - x z') \mathbf{j} + (x y' - y x') \mathbf{k},
\end{gather*}
founded the vector calculus notation of the scalar and cross product, respectively, which was then most prominently developed and promoted by Heaviside (see e.g.\ \cite{Hunt2012} and the references therein) and by Gibbs \cite{Gibbs1884, GibbsWilson1901} who coined most of the notation
\begin{align*}
   &
   \nabla f
   \coloneqq
   \operatorname{grad} f
   \coloneqq
   \big( \partial_1 f, \partial_2 f, \partial_3 f \big)
   \qquad \text{for } f \in C^1(\R^3, \R),
\\[0.8ex]
   &
   \nabla \times \mathbf a
   \coloneqq
   \operatorname{curl} \mathbf a
   \coloneqq
   \big(
   \partial_2 \mathbf a_3 - \partial_3 \mathbf a_2, 
   \partial_3 \mathbf a_1 - \partial_1 \mathbf a_3, 
   \partial_1 \mathbf a_2 - \partial_2 \mathbf a_1 
   \big),
\\[0.7ex]
   &
   \nabla \cdot \mathbf a
   \coloneqq
   \operatorname{div} \mathbf a
   \coloneqq
   \partial_1 \mathbf a_1 + \partial_2 \mathbf a_2 + \partial_3 \mathbf a_3
   \qquad \text{for } \mathbf a \in C^1(\R^3, \R^3),
\end{align*}
which is still used today in vector calculus.
\hfill $\diamond$
\end{historicalremark}

Heaviside developed vector analysis independently of Gibbs and it was him who expressed Maxwell's equations in vector calculus notation (see e.g.\ \cite{Arthur2013, Unz1963} and the many references therein). 
\begin{figure}[h!]
\small
\begin{tcolorbox}[ams align*]
   \left.
   \begin{aligned}
      a &= \tfrac{dH}{dy} - \tfrac{dG}{dz}
   \\
      b &= \tfrac{dF}{dz} - \tfrac{dH}{dx}
   \\
      c &= \tfrac{dG}{dx} - \tfrac{dF}{dy}
   \end{aligned}
   \right\}
   & \quad \mathrm{(A)}
   & \mathbf{B} &= \nabla \times \mathbf{A}
\\[0.2ex]
   \left.
   \begin{aligned}
      P &= c \tfrac{dy}{dt} - b \tfrac{dz}{dt} - \tfrac{dF}{dt} - \tfrac{d\Psi}{dx}
   \\
      Q &= a \tfrac{dz}{dt} - c \tfrac{dx}{dt} - \tfrac{dG}{dt} - \tfrac{d\Psi}{dy}
   \\
      R &= b \tfrac{dx}{dt} - a \tfrac{dy}{dt} - \tfrac{dH}{dt} - \tfrac{d\Psi}{dz}
   \end{aligned}
   \right\}
   & \quad \mathrm{(B)}
   & \mathbf{E} &= \mathbf{v} \times \mathbf{B} - \tfrac{\partial \mathbf{A}}{\partial t} - \nabla \phi
\\[0.2ex]
   \left.
   \begin{aligned}
      X &= v c - w b
   \\
      Y &= w a - u c
   \\
      Z &= u b - v a
   \end{aligned}
   \right\}
   & \quad \mathrm{(C)}
   & \mathbf{F} &= \mathbf{J} \times \mathbf{B}
\\[0.2ex]
   \left.
   \begin{aligned}
      a &= \alpha + 4 \pi A
   \\
      b &= \beta + 4 \pi B
   \\
      c &= \gamma + 4 \pi C
   \end{aligned}
   \right\}
   & \quad \mathrm{(D)}
   & \mathbf{B} &= \mu_0 \mathbf{H} + \mathbf{M}
\\[0.2ex]
   \left.
   \begin{aligned}
      4 \pi u &= \tfrac{d\gamma}{dy} - \tfrac{d\beta}{dz}
   \\
      4 \pi v &= \tfrac{d\alpha}{dz} - \tfrac{d\gamma}{dx}
   \\
      4 \pi w &= \tfrac{d\beta}{dx} - \tfrac{d\alpha}{dy}
   \end{aligned}
   \right\}
   & \quad \mathrm{(E)}
   & \mathbf{J} &= \nabla \times \mathbf{H}
\\[0.2ex]
   \mathfrak{D} = \tfrac{1}{4\pi} \mathfrak{E}
   \hspace*{2.3ex}
   & \quad \mathrm{(F)}
   & \mathbf{D} &= \epsilon \mathbf{E}
\\[0.2ex]
   \mathfrak{K} = C \mathfrak{E}
   \hspace*{2.3ex}
   & \quad \mathrm{(G)}
   & \mathbf{J}_c &= \sigma \mathbf{E}
\\[-0.1ex]
   \mathfrak{C} = \mathfrak{K} + \mathfrak{\dot D}
   \hspace*{2.3ex}
   & \quad \mathrm{(H)}
   & \mathbf{J} &= \mathbf{J}_c + \tfrac{\partial \mathbf{D}}{\partial t}
\\[0.1ex]
   \left.
   \begin{aligned}
      u &= p + \tfrac{d f}{d t}
   \\
      v &= q + \tfrac{d g}{d t}
   \\
      w &= r + \tfrac{d h}{d t}
   \end{aligned}
   \right\}
   & \quad \mathrm{(H^*)}
   & \mathbf{J} &= \mathbf{J}_c + \tfrac{\partial \mathbf{D}}{\partial t}
\\[0.2ex]
   \mathfrak{C} = (C + \tfrac{1}{4 \pi} K \tfrac{d}{dt}) \mathfrak{E}
   \hspace*{2.3ex}
   & \quad \mathrm{(I)}
   & \mathbf{J} &= \sigma \mathbf{E} + \epsilon \tfrac{\partial \mathbf{E}}{\partial t}
\\[0.1ex]
   \left.
   \begin{aligned}
      u &= C P + \tfrac{1}{4 \pi} K \tfrac{d P}{d t}
   \\
      v &= C Q + \tfrac{1}{4 \pi} K \tfrac{d Q}{d t}
   \\
      w &= C R + \tfrac{1}{4 \pi} K \tfrac{d R}{d t}
   \end{aligned}
   \right\}
   & \quad \mathrm{(I^*)}
   & \mathbf{J} &= \sigma \mathbf{E} + \epsilon \tfrac{\partial \mathbf{E}}{\partial t}
\\[0.2ex]
   \rho = \tfrac{df}{dx} + \tfrac{dg}{dy} +\tfrac{dh}{dz}
   \hspace*{2.3ex}
   & \quad \mathrm{(J)}
   & \varrho &= \nabla \cdot \mathbf{D}
\\[0.2ex]
   \sigma = l f + m g + n h + l' f' + m' g' + n' h'
   \hspace*{2.3ex}
   & \quad \mathrm{(K)}
   & \varrho_s &= \mathbf{n} \cdot (\mathbf{D_1} - \mathbf{D_2})
\\[0.2ex]
   \mathfrak{B} = \mu \mathfrak{H}
   \hspace*{2.3ex}
   & \quad \mathrm{(L)}
   & \mathbf{B} &= \mu \mathbf{H}
\end{tcolorbox}
\vspace*{-2ex}
\caption{Maxwell's original equations.\label{fig:MaxwellsEquationsOriginal}}
\end{figure}
Figure \ref{fig:MaxwellsEquationsOriginal} displays Maxwell's original set of equations \cite[vol.\ 2, pp.\ 215-233]{Maxwell1873} which were almost completely in coordinate-wise ``longhand'' notation, together with their interpretation in modern Gibbsean vector calculus notation (cp.\ also \cite[p.\ 3]{Lindel2004}).

\section{Vector analysis, integration theorems and\newline Poincar\'{e}'s Lemma} \label{S:label1}

One of the cornerstones of vector analysis is the fact that the sequence
\begin{equation*}
   0 
   \to
   \R
   \to
   C^\infty(\Omega, \R)
   \stackrel{\operatorname{grad}}{\to}
   C^\infty(\Omega, \R^3)
   \stackrel{\operatorname{curl}}{\to}
   C^\infty(\Omega, \R^3)
   \stackrel{\operatorname{div}}{\to}
   C^\infty(\Omega, \R)
   \to
   0
\end{equation*}
is exact. We work with the following more detailed description of this statement.

\begin{theorem}[Poincar\'{e}'s Lemma]\label{thm:poincare-lemma}
Let $\Omega \subset \R^3$ be open and star shaped with star center $x_0 \in \Omega$.
\begin{itemize}
\item[(a)] $\operatorname{im}(\operatorname{grad}) = \operatorname{ker}(\operatorname{curl})$, more precisely, for each $\mathbf{a} \in C^1(\Omega, \R^3)$
\begin{equation*}
   \operatorname{curl} \mathbf{a} = 0
   \qquad
   \Leftrightarrow
   \qquad
   \exists f \in C^1(\Omega, \R)
   \colon
   \mathbf{a} = \operatorname{grad} f,
\end{equation*}
e.g.\ $f(\mathbf{x}) \coloneqq \int_0^1\langle \mathbf{a}(\mathbf{x}_0 + t(\mathbf{x}-\mathbf{x}_0)), \mathbf{x}-\mathbf{x}_0\rangle dt$ for $\mathbf{x}\in\Omega$.

\bigskip
\item[(b)] $\operatorname{im}(\operatorname{curl}) = \operatorname{ker}(\operatorname{div})$, more precisely, for each $\mathbf{b} \in C^1(\Omega, \R^3)$
\begin{equation*}
   \operatorname{div} \mathbf{b} = 0
   \qquad
   \Leftrightarrow
   \qquad
   \exists \mathbf{a} \in C^1(\Omega, \R^3)
   \colon
   \mathbf{b} = \operatorname{curl} \mathbf{a},
\end{equation*}
e.g.\ $\mathbf{a}(\mathbf{x}) \coloneqq \int_0^1 t \mathbf{b}(\mathbf{x}_0 + t(\mathbf{x} - \mathbf{x}_0)) d t \times (\mathbf{x}-\mathbf{x}_0)$ for $\mathbf{x} \in \Omega$.

\bigskip
\item[(c)] $\operatorname{im}(\operatorname{div}) = C^0(\Omega, \R)$, more precisely, for each $g \in C^0(\Omega, \R)$
\begin{equation*}
   \exists \mathbf{b} \in C^1(\Omega, \R^3)
   \colon
   g = \operatorname{div} \mathbf{b},
\end{equation*}
e.g.\ $\mathbf{b}(\mathbf{x}) \coloneqq \int_0^1 t^2 g(\mathbf{x}_0+t(\mathbf{x}-\mathbf{x}_0))(\mathbf{x}-\mathbf{x}_0) dt$ for $\mathbf{x}\in\Omega$.
\end{itemize}
\end{theorem}

\begin{proof}
Theorem \ref{thm:poincare-lemma} is a special case of Poincar\'{e}'s Lemma for differential forms (see e.g.\ \cite[Theorems 11.49 \& 17.14 and Lemma 17.27]{Lee2012}). We prove only (a). The proofs of (b) and (c) have the same structure due to the fact that the exterior derivative of a differential form specializes to $\operatorname{div}$, $\operatorname{curl}$ and $\operatorname{grad}$, respectively.
\medskip

(a) The implication $(\Leftarrow)$ follows easily by computing $\operatorname{curl} \operatorname{grad} f = 0$. To show the implication $(\Rightarrow)$ define $\mathbf{w}(\mathbf{x})\coloneqq \mathbf{x}-\mathbf{x}_0$, $(w_1, w_2, w_3) \coloneqq \mathbf{w}$, $(a_1, a_2, a_3) \coloneqq \mathbf{a}$ and $f(\mathbf{x}) \coloneqq \int_0^1\sum_{k=1}^3 a_k(\mathbf{x}_0+t\mathbf{w}(\mathbf{x})) w_k(\mathbf{x}) dt$ for $\mathbf{x}\in\Omega$.
Let $j\in\{1,2,3\}$.
We use $\partial_j a_k = \partial_k a_j$ for $k\in\{1,2,3\}$ to compute
\begin{align*}
        \partial_jf(\cdot) &= \int_0^1\sum_{k=1}^3\partial_j\big(a_k(\mathbf{x}_0+t\mathbf{w}(\cdot))w_k(\cdot)\big) dt
        \\
        &= \int_0^1\Big(\sum_{k=1}^3\big\langle\operatorname{grad}a_k(\mathbf{x}_0+t\mathbf{w}(\cdot)),t\mathbf{e}_j\big\rangle w_k(\cdot) + a_j(\mathbf{x}_0+t\mathbf{w}(\cdot))\Big) dt
        \\
        &= \int_0^1 t\sum_{k=1}^3\partial_j a_k(\mathbf{x}_0+t\mathbf{w}(\cdot)) w_k(\cdot) dt + \int_0^1 a_j(\mathbf{x}_0+t\mathbf{w}(\cdot)) dt
        \\
        &= \int_0^1 t\sum_{k=1}^3\partial_k a_j(\mathbf{x}_0+t\mathbf{w}(\cdot)) w_k(\cdot) dt + \int_0^1 a_j(\mathbf{x}_0+t\mathbf{w}(\cdot)) dt
        \\
        &= \int_0^1 t\big\langle\operatorname{grad} a_j(\mathbf{x}_0+t\mathbf{w}(\cdot)),\mathbf{w}(\cdot)\big\rangle dt + \int_0^1 a_j(\mathbf{x}_0+t\mathbf{w}(\cdot)) dt.
\end{align*}
However, for $\mathbf{x}\in\Omega$
\begin{equation*}
        t\mapsto\big\langle\operatorname{grad} a_j(\mathbf{x}_0+t\mathbf{w}(\mathbf{x})),\mathbf{w}(\mathbf{x})\big\rangle = \frac{\mathrm d}{\mathrm dt}\big(t\mapsto a_j(\mathbf{x}_0+t\mathbf{w}(\mathbf{x}))\big).
\end{equation*}
We conclude by partial integration that for $\mathbf{x}\in\Omega$
\begin{align*}
        \partial_jf(\mathbf{x}) &= t a_j(\mathbf{x}_0+t\mathbf{w}(\mathbf{x}))\Big\vert_{t=0}^1 - \int_0^1 a_j(\mathbf{x}_0+t\mathbf{w}(\mathbf{x})) dt
        \\
        &\qquad + \int_0^1 a_j(\mathbf{x}_0+t\mathbf{w}(\mathbf{x})) dt = a_j(\mathbf{x}).
        \qedhere
\end{align*}
\end{proof}

\begin{historicalremark}[Early Days of Electrodynamics]
The early days of \emph{electrodynamics} - a term coined by Ampère \cite{Erlichson1998} - read like an exciting detective novel and are excellently presented e.g.\ in \cite{Darrigol2003, HonGoldstein2020, Siegel1991} and the references therein. The following timeline is adapted from \cite{Turnbull2013}.

\begin{tabular}{rl} 
\\
\toprule
   1785 & Coulomb's Law is published
\\
   1813 & Gauss's Divergence Theorem is discovered
\\
   1820 & {\O}rsted discovers that an electric current creates a magnetic field
\\
   1820 & Amp\`{e}re's work founds electrodynamics; Biot-Savart Law is discovered
\\
   1826 & Amp\`{e}re's Formula is published
\\
   1831 & Faraday's Law is published
\\
   1856 & Maxwell publishes ``On Faraday's lines of force''
\\
   1861 & Maxwell publishes ``On physical lines of force''
\\
   1865 & Maxwell publishes ``A dynamical theory of the electromagnetic field''
\\
   1873 & Maxwell publishes ``A Treatise on Electricity and Magnetism''
\\
   1888 & Hertz discovers radio waves
\\
   1940 & Einstein popularizes the name ``Maxwell's Equations''
\\ \bottomrule
\end{tabular}
\medskip

In 1785 Coulomb published his law, which states that the force between two electrical charges is proportional to the product of the charges and inversely proportional to the square of the distance between them (see e.g.\ \cite{YoutubeCoulomb} for a docu\-mentary). 

In 1813 Gauss discovered a special case of the divergence theorem. Todays version, that the surface integral of a vector field over a closed surface is equal to the volume integral of the divergence over the region inside the surface, was formulated and proved by Ostrogradsky (not by Gauss) 13 years later. Maxwell credits the divergence theorem to Ostrogradsky. He cited a paper by Ostrogradsky from the correct year but with a wrong title. This citation was later removed but can still be found in the first edition of Maxwell’s Treatise on Electricity and Magnetism.

In 1820 {\O}rsted discovered that a compass needle was deflected from the magnetic north direction by a nearby electric current. He thus confirmed a direct relationship between electricity and magnetism.  Already on September 30, 1820, Biot and Savart announced their results for the distance dependence of the magnetic force exerted by a long, straight current-carrying wire on a magnetic needle. 

In 1826 Ampère’s formula on the force between two infinitesimal current elements was published. Maxwell respected Ampère’s work and his genius but he questioned the action-at-a-distance concept which was commonly accepted then as the state-of-the-art also by Ampère e.g.\ in his formula. When Maxwell later developed his field theory of electromagnetism Ampère’s formula fell into disuse. Sometime around the 1930s Ampère’s name came to be associated with ‘‘the first law of circulation’’, then also called Ampère’s Law which, however, is not due to Ampère. One should therefore speak of Ampère’s formula which was published in 1826.

Faraday was one of the most influential scientists in history and definitely during his time. He was an excellent experimentalist who conveyed his ideas in clear and simple language (see e.g.\ \cite{YoutubeFaraday} for a documentary). He published from 1831 until 1854 in the “Philosophical Transactions of the Royal Society“ a series of articles with the title “Experimental Researches in Electricity“. He chose this also as the title of a three volume book in which he summarised his research results on electricity in 30 series of articles consisting of 3430 numbered paragraphs.
Maxwell acknowledged and accepted the phenomena discovered by {\O}rsted, Amp\`{e}re, and others but, principally, he depended on the conceptual framework and experimental achievements of Faraday \cite[p.\ 4]{HonGoldstein2020} (see \cite{Faraday} for a reprint of Faraday's Experimental Researches in Electricity). In the preface to his Treatise on Electricity and Magnetism of 1873 [16, p.\ viii], Maxwell recalled:

\vspace*{-0.2ex}
\begin{quote}\footnotesize
[B]efore I began the study of electricity I resolved to read no mathematics on the subject till I had first read through Faraday's Experimental Researches in Electricity.
\end{quote}
\vspace*{-0.3ex}

Maxwell's initial publication \cite{Maxwell1856} of 1856 (an abstract) on electromagnetism placed his methodology at the forefront \cite[Sec.\ 1.3, p.\ 9]{HonGoldstein2020}. For details we recommend the excellent discussion of the importance of methodology for Maxwell in Hon and Goldstein \cite[Secs.\ 1.3 \& 1.4]{HonGoldstein2020} from which we quote freely. At the beginning of the abstract \cite{Maxwell1856} Maxwell indicated that the methodology he adopted was a modified version of the formal analogy invoked by Thomson \cite[`On the uniform motion of heat in homogeneous solid bodies, and its connexion with the mathematical theory of electricity', pp.\ 1--14]{Thomson2011}:

\vspace*{-0.2ex}
\begin{quote}\footnotesize
The method pursued in this paper is a modification of that mode of viewing electrical phenomena in relation to the theory of the uniform conduction of heat, which was first pointed out by Professor W.\ Thomson, [...] Instead of using the analogy of heat, a fluid, the properties of which are entirely at our disposal [i.e., purely imaginary], is assumed as the vehicle of mathematical reasoning.
\end{quote}
\vspace*{-0.3ex}

Maxwell actually stated that he would modify Thomson’s methodology of ana\-logy that relates two distinct physical domains via the same mathematical structure. He used Faraday’s concept of lines of force as processes of reasoning and introduced an imaginary imponderable and incompressible fluid that permeates a medium whose resistance is directly proportional to the velocity of the fluid. 
\hfill $\diamond$
\end{historicalremark}

We now recall the classical integration theorems, called gradient theorem, curl theorem (theorem of Stokes or Kelvin-Stokes theorem) and divergence theorem (theorem of Gauß or Ostrogradsky-Gauß theorem), and fix the notation for the integration of a function $f \colon M \to \R$ on a $k$-dimensional submanifold $M \subset \R^n$. For simplicity assume that $\Phi \colon U \to V$, with $U \subset \R^k$ open, is a chart of $M$ such that $\operatorname{supp} f \subset V$. Then with the Gramian $g(\mathbf{u}) \coloneqq \det(D\Phi(\mathbf{u})^\top D\Phi(\mathbf{u}))$ the integral is defined by
\begin{equation*}
   \int_M f(\mathbf{x}) \,dS(\mathbf{x})
   \coloneqq
   \int_U f(\Phi(\mathbf{u}))\sqrt{g(\mathbf{u})} \,d \mathbf{u}.
\end{equation*}
We follow common physics notation, omit the argument $\mathbf{x}$, i.e.\ we write $\int_M f \,dS$, and for $n=3$ notationally distinguish between the cases $k=1,2,3$ by writing $ds$, $dS$ and $dV$ instead of $dS$, respectively. Moreover, to further bridge the notational gap, we typeset vectorial line and surface elements in bold face as displayed and defined in the following theorem.

\begin{theorem}[Classical Integration Theorems]\label{thm:integration}
Let $\Omega \subset \R^3$ be open.
\medskip

\textbf{Gradient Theorem.}
Let $\gamma \colon [0,1] \to \Omega$ be piecewise continuously differentiable and $f \in C^1(\Omega, \R)$. Then
\begin{equation*}
   \int_0^1 \langle \operatorname{grad} f(\gamma(t)), \gamma'(t)\rangle dt
   =
   f(\gamma(1)) - f(\gamma(0)),
\end{equation*}
and with $\mathbf{p} = \gamma(0)$, $\mathbf{q} = \gamma(1)$, even shorter and more common in physics literature
\begin{equation*}
   \int_\gamma \operatorname{grad} f \cdot d\mathbf{s}
   =
   f(\mathbf{q}) - f(\mathbf{p}).
\end{equation*}
\medskip

\textbf{Curl Theorem.}
Let $M \subset \Omega$ be a bounded piecewise smooth oriented two-dimensional manifold with unit normal field $\nu : M \to \R^3$ and piecewise smooth boundary curve $\partial M$ with induced unit tangent field $\tau : \partial M \to \R^3$, and $\mathbf{a} \in C^1(\Omega, \R^3)$. Then
\begin{equation*}
   \int_M \langle \operatorname{curl} \mathbf{a}, \nu \rangle dS
   =
   \int_{\partial M} \langle \mathbf{a}, \tau \rangle ds,
\end{equation*}
and even shorter and more common in physics literature
\begin{equation*}
   \int_M \operatorname{curl} \mathbf{a} \cdot d\mathbf S
   =
   \int_{\partial M} \mathbf{a} \cdot d\mathbf s.
\end{equation*}
\medskip

\textbf{Divergence Theorem.}
Let $M \subset \Omega$ be bounded open with piecewise smooth boundary $\partial M$, $\nu : \partial M \to \R^3$ the outer unit normal field, and $\mathbf{b} \in C^1(\Omega, \R^3)$. Then
\begin{equation*}
   \int_M \operatorname{div} \mathbf{b} \,dV
   =
   \int_{\partial M} \langle \mathbf{b}, \nu \rangle \,dS,
\end{equation*}
and even shorter and more common in physics literature
\begin{equation*}
   \int_M \operatorname{div} \mathbf{b} \,dV
   =
   \int_{\partial M} \mathbf{b} \cdot d\mathbf S.
\end{equation*}
\end{theorem}

\begin{proof}
The proofs of these classical integration theorems can be found in many textbooks. We refer to \cite[Chapter 8, p.\ 261, Stokes' Theorem]{Spivak1970} for a unified proof of all three statements using differential forms.
\end{proof}

Whereas the Classical Integration Theorems \ref{thm:integration} allow for an integral interpretation of the differential operators $\operatorname{div}$, $\operatorname{curl}$ and $\operatorname{grad}$, the following theorem additionally also provides the converse implication.  

\begin{theorem}[Geometric interpretation of gradient, curl and divergence]\label{thm:geometric-vectorcalculus}
Let $\Omega \subset \R^3$ be non-empty and open.
\begin{itemize}
\item[(a)] Let $\Gamma \coloneqq \{ \gamma \colon [0,1] \to \Omega \,|\, \gamma \text{ is piecewise continuously differentiable}\}$. Then for each $f \in C^1(\Omega, \R)$ and $\mathbf{a} \in C^0(\Omega, \R^3)$ the following statements are equivalent:
\vspace*{-0.3ex}
\begin{itemize}
   \item[(i)] $\forall \gamma \in \Gamma
   \colon
   \int_\gamma \mathbf{a} \cdot d\mathbf{s} = f(\gamma(1)) - f(\gamma(0))$.

   \item[(ii)] $\mathbf{a} = \operatorname{grad} f$.
\end{itemize}

\item[(b)] Let $\mathcal{M}$ be the set of bounded oriented piecewise smooth two-dimensional submanifolds $M$ of $\Omega$ with piecewise smooth boundary $\partial M$. Then for each $\mathbf{a} \in C^1(\Omega, \R^3)$ and $\mathbf{b} \in C^0(\Omega, \R^3)$ the following statements are equivalent:
\vspace*{-0.3ex}
\begin{itemize}
   \item[(i)] $\forall M \in \mathcal{M}
   \colon
   \int_M \mathbf{b} \cdot d\mathbf{S} = \int_{\partial M} \mathbf{a} \cdot d\mathbf{s}$.

   \item[(ii)] $\mathbf{b} = \operatorname{curl} \mathbf{a}$.
\end{itemize}

\item[(c)] 
Let $\mathcal{M}$ be the set of bounded open subsets $M$ of $\Omega$ with piecewise smooth boundary. Then for each $\mathbf{b} \in C^1(\Omega, \R^3)$ and $g \in C^0(\Omega, \R)$ the following statements are equivalent:
\vspace*{-0.3ex}
\begin{itemize}
   \item[(i)] $\forall M \in \mathcal{M}
   \colon
   \int_M g \,dV = \int_{\partial M} \mathbf{b} \cdot d\mathbf{S}$.

   \item[(ii)] $g = \operatorname{div} \mathbf{b}$.
\end{itemize}
\end{itemize}
\end{theorem}

For the proof of Theorem \ref{thm:geometric-vectorcalculus} we use the following Lemma.

\begin{lemma}\label{lem:det_by_int}
Let \(\Omega \subset \mathbb R^3\) be non-empty and open.
Let \(\mathbf v \in C^0(\Omega,\mathbb R^3)\) and \(g \in C^0(\Omega,\mathbb R)\).
The following statements hold:
\begin{itemize}
    \item[(a)]
    If for all paths \(\gamma \colon [0,1] \to \Omega,\, t \mapsto (1-t)\mathbf{x}_0 + t\mathbf{x}_1\) where \(\mathbf{x}_0,\mathbf{x}_1 \in \Omega\), \(\int_\gamma \mathbf v\cdot d\mathbf s = 0\), then \(\mathbf v = 0\).
\\[-1.6ex]
    \item[(b)]
    If for all discs \(D(\mathbf{x},\mathbf{n},r) \subset \Omega\) with center \(\mathbf{x} \in \Omega\), radius \(r > 0\), unit normal \(\mathbf{n} \in \R^3\) and the orientation induced by \(\mathbf{n}\), \(\int_{D(\mathbf{x},\mathbf{n},r)} \mathbf v\cdot d\mathbf S = 0\), then \(\mathbf v = 0\).
\\[-1.6ex]
    \item[(c)]
    If for all balls \(B(\mathbf{x},r) \subset \Omega\) with center \(\mathbf{x} \in \Omega\) and radius \(r > 0\),\linebreak \(\int_{B(\mathbf{x},r)} g\,dV = 0\), then \(g = 0\).
\end{itemize}
\end{lemma}

\begin{proof}
We only prove (b).
Suppose to the contrary, that \(\mathbf v \neq 0\).
Then there is \(\mathbf x \in \Omega\) with \(\mathbf v(\mathbf x) \neq 0\).
We define \(\mathbf n \coloneqq \mathbf v(\mathbf x)/\Vert \mathbf v(\mathbf x)\Vert\).
Then the function \(g(\mathbf y) \coloneqq \langle \mathbf v(\mathbf y),\mathbf n\rangle\) is continuous in \(\Omega\) and satisfies \(g(\mathbf x) = \Vert \mathbf v(\mathbf x)\Vert > 0\).
Hence there is \(r > 0\), s.t.\ \(B(\mathbf x,r) \subseteq \Omega\) and for all \(\mathbf y \in B(\mathbf x,r)\) it holds that \(g(\mathbf y) > 0\).
Then the disk \(D(\mathbf x,\mathbf n,r) \coloneqq \{\mathbf y \in B(\mathbf x,r) \,\vert\, \langle \mathbf y-\mathbf x,\mathbf n\rangle = 0\}\) has unit normal field \(D(\mathbf x,\mathbf n,r) \ni \mathbf y \mapsto \mathbf n\) and can be parameterized almost globally by a parameterization \(\gamma:\dom\gamma\subset\R^2 \to D(\mathbf x,\mathbf n,r)\), where \(\dom\gamma\) has non-zero measure.
We compute
\begin{equation*}
    0 = \int_{D(\mathbf x,\mathbf n,r)} \mathbf v\cdot d\mathbf S = \int_{\dom\gamma}\langle \mathbf v(\gamma(s,t)),\mathbf n\rangle \,d(s,t) = \int_{\dom\gamma} g(\gamma(s,t)) \,d(s,t) > 0,
\end{equation*}
which is a contradiction.
\end{proof}

\begin{proof}[Proof of Theorem \ref{thm:geometric-vectorcalculus}]
The proposition \(\text{(a)}\textrm{(ii)}\Rightarrow\textrm{(i)}\) follows by an application of the Fundamental Theorem of Calculus and \(\text{(b)}\textrm{(ii)}\Rightarrow\textrm{(i)}\) respectively \(\text{(c)}\textrm{(ii)}\Rightarrow\textrm{(i)}\) follows from the Curl respectively the Divergence Theorem.
The converse implications \(\text{(a)}\textrm{(i)}\Rightarrow\textrm{(ii)}\), \(\text{(b)}\textrm{(i)}\Rightarrow\textrm{(ii)}\) and \(\text{(c)}\textrm{(i)}\Rightarrow\textrm{(ii)}\) can be deduced again from the Fundamental Theorem of Calculus, the Curl and the Divergence Theorem and Lemma \ref{lem:det_by_int}.
We provide the proof only for \(\text{(b)}\textrm{(i)}\Rightarrow\textrm{(ii)}\).
By an application of the Curl Theorem, we compute for \(M \in \mathcal M\),
\begin{equation*}
    \int_{M} \big(\mathbf b - \operatorname{curl}\mathbf a\big)\cdot d\mathbf S = 0.
\end{equation*}
Since all discs \(D(\mathbf{x},\mathbf{n},r) \subseteq \Omega\), with center \(\mathbf{x} \in \mathbb R^3\), unit normal \(\mathbf{n} \in \mathbb R^3\) and radius \(r > 0\) are elements of \(\mathcal M\), we deduce from Lemma \ref{lem:det_by_int} that \(\mathbf b - \operatorname{curl} \mathbf a = 0\).
\end{proof}

Combining Poincar\'{e}'s Lemma \ref{thm:poincare-lemma} with the geometric interpretation of gradient, curl and divergence in Theorem \ref{thm:geometric-vectorcalculus}, we obtain our main tool for the construction of Maxwell's equations.

\begin{lemma}[Geometric interpretation of Poincar\'{e}'s Lemma]\label{lem:geometric-poincare-lemma}
Let $\Omega \subset \R^3$ be open and star shaped.
\begin{itemize}
\item[(a)] For each $\mathbf{a} \in C^1(\Omega, \R^3)$
\begin{equation*}
   \operatorname{curl} \mathbf{a} = 0
   \quad
   \Leftrightarrow
   \quad
   \exists f \in C^1(\Omega, \R)
   \colon
   \text{\emph{Theorem \ref{thm:geometric-vectorcalculus}(a)(i) and (a)(ii) hold}}.
\end{equation*}

\item[(b)] For each $\mathbf{b} \in C^1(\Omega, \R^3)$
\begin{equation*}
   \operatorname{div} \mathbf{b} = 0
   \quad
   \Leftrightarrow
   \quad
   \exists \mathbf{a} \in C^1(\Omega, \R^3)
   \colon
   \text{\emph{Theorem \ref{thm:geometric-vectorcalculus}(b)(i) and (b)(ii) hold}}.
\end{equation*}

\item[(c)] For each $g \in C^0(\Omega, \R)$
\begin{equation*}
   \exists \mathbf{b} \in C^1(\Omega, \R^3)
   \colon
   \text{\emph{Theorem \ref{thm:geometric-vectorcalculus}(c)(i) and (c)(ii) hold}}.
\end{equation*}
\end{itemize}
\end{lemma}

\section{Maxwell's Imaginary Fluid}

In this section we model Maxwell's imaginary fluid based on his original work and applications of the geometric interpretation of Poincar\'{e}'s Lemma \ref{lem:geometric-poincare-lemma}. All occurring quantities are assumed to be smooth enough for their integrals and derivatives to exist as noted.

\begin{historicalremark}[Maxwell's Imaginary Fluid and Tube of Fluid Motion]\label{rem:tube}
In \cite{Maxwell1858} Maxwell introduces a hypothetical fluid as a purely imaginary substance with a collection of imaginary properties. We quote from \cite[pp.\ 160--162]{MaxwellNiven1890}.

\begin{quote}\footnotesize
   (1) The substance here treated of must not be assumed to possess any of the properties of ordinary fluids 
   except those of freedom of motion and resistance to compression.   [$\dots$]
   
   \emph{The portion of fluid which at any instant occupied a given volume, will at any succeeding instant occupy an equal volume.}
   
   This law expresses the incompressibility of the fluid, and furnishes us with a convenient measure of its quantity, 
   namely its volume. The unit of quantity of the fluid will therefore be the unit of volume.
\end{quote}

\begin{quote}\footnotesize
   (2) The direction of motion of the fluid will in general be different at different points of the space which it occupies, but since the direction is determinate for every such point, we may conceive a line to begin at any point and to be continued so that every element of the line indicates by its direction the direction of motion at that point of space. Lines drawn in such a manner that their direction always indicates the direction of fluid motion are called lines of fluid motion. [$\dots$]
\end{quote}

\vspace*{-1ex}
\begin{figure}[!h]
\hspace*{5ex}
\includegraphics*[width=0.3\textwidth]{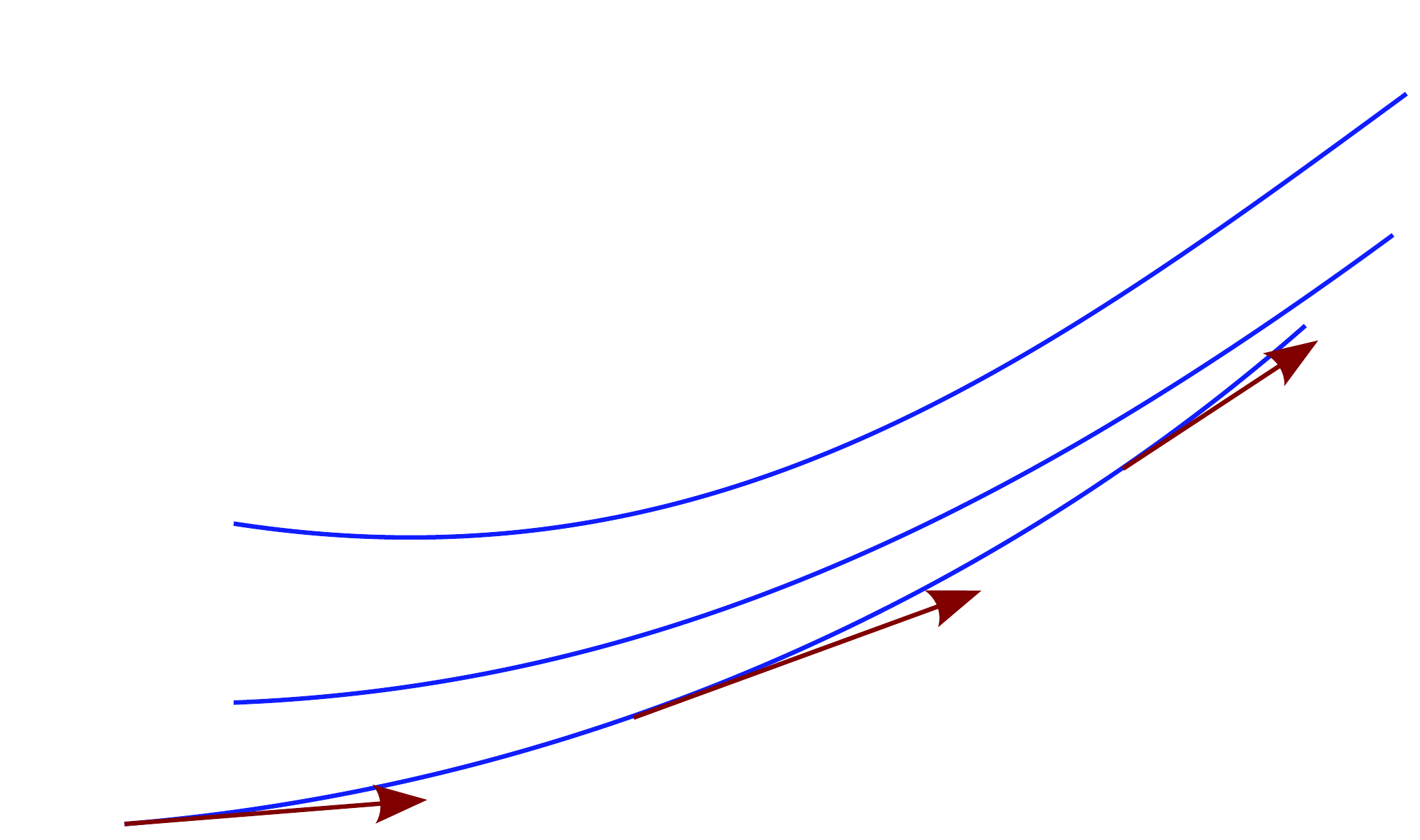}
\hspace*{9ex}
\includegraphics*[width=0.4\textwidth]{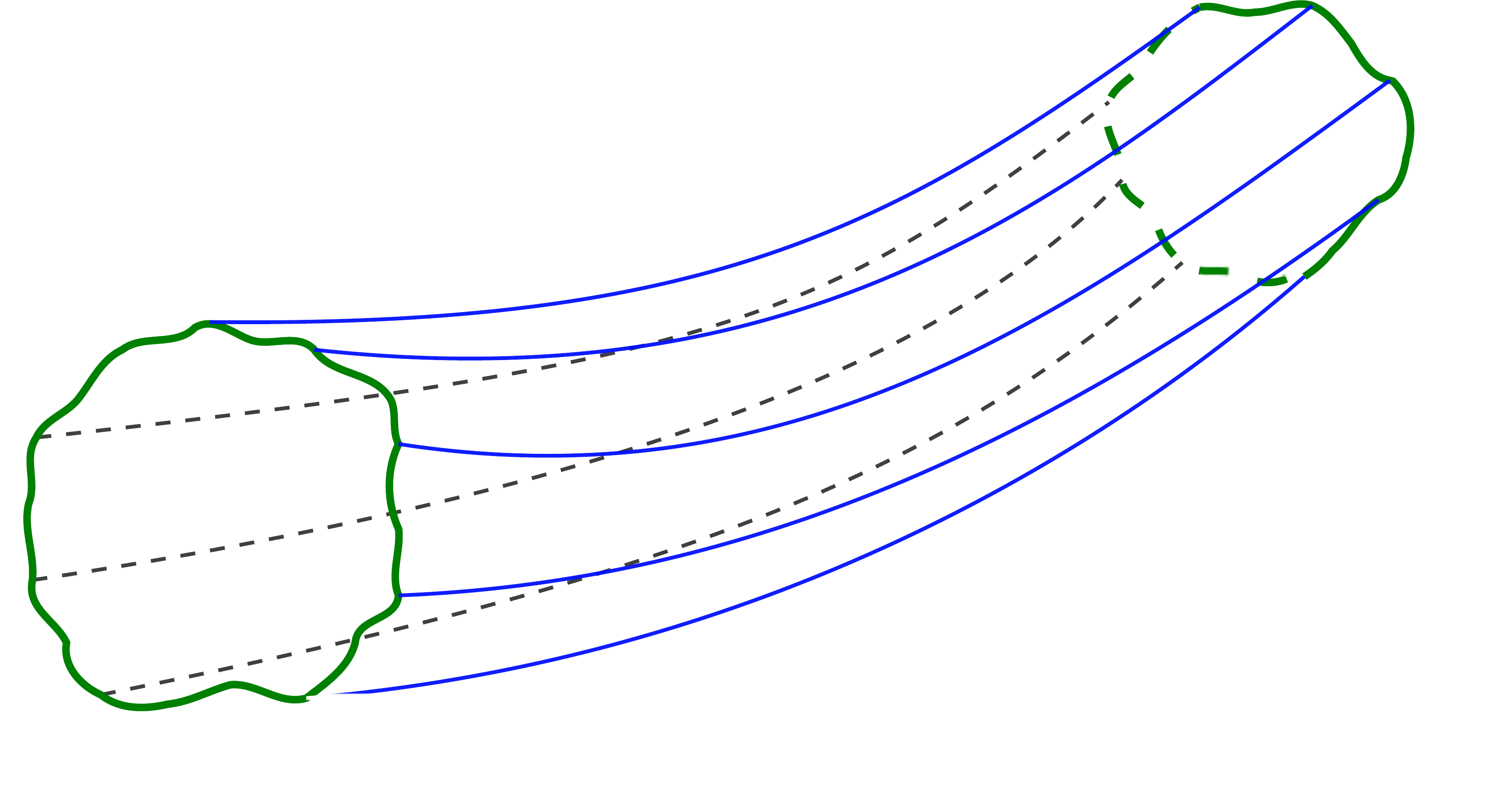}
\caption{Lines of fluid motion and tube of fluid motion.\label{fig:fluidmotion}}
\end{figure}

\begin{quote}\footnotesize
(3) If upon any surface which cuts the lines of fluid motion we draw a closed curve, and if from every point of this curve we draw a line of motion, these lines of motion will generate a tubular surface which we may call a tube of fluid motion. Since this surface is generated by lines in the direction of fluid motion no part of the fluid can flow across it, so that this imaginary surface is as impermeable to the fluid as a real tube.
\end{quote}

\begin{quote}\footnotesize
(7) [$\dots$] if the origin of the tube or its termination be within the space under consideration, then we must conceive the fluid to be supplied by a source within that space, capable of creating and emitting unity of fluid in unity of time, and to be afterwards swallowed up by a sink capable of receiving and destroying the same amount continually. [$\dots$]
\end{quote}
\vspace*{-3.3ex}\hfill $\diamond$
\end{historicalremark}

The imaginary fluid which Maxwell introduces in property (1) of the Historical Remark \ref{rem:tube} occupies volume in a domain $\Omega \subset \R^3$ and has, in particular, a velocity depending on time $t$ and described by a velocity field
\begin{equation*}
   \mathbf{v} : \R \times \Omega \to \R^3,
   \quad
   (t,\mathbf{x}) \mapsto \mathbf{v}(t,\mathbf{x}).
\end{equation*}
Let $S$ be an oriented surface, more precisely, a bounded piecewise smooth oriented two-dimensional manifold in $\Omega$. Then 
\begin{equation*}
   \int_{S} \mathbf{v}(t,\cdot) \cdot d \mathbf{S}
\end{equation*}
describes the flow at time $t$ through the surface $S$, also called the \emph{flux} through $S$. With this interpretation we rename $\mathbf{v}$ and call it from now on \emph{(fluid) flux density}, see Figure \ref{fig:flux}.

\begin{figure}[!h]\label{fig:flux}
\vspace*{-1ex}\hspace*{8ex}
\includegraphics*[width=0.7\textwidth]{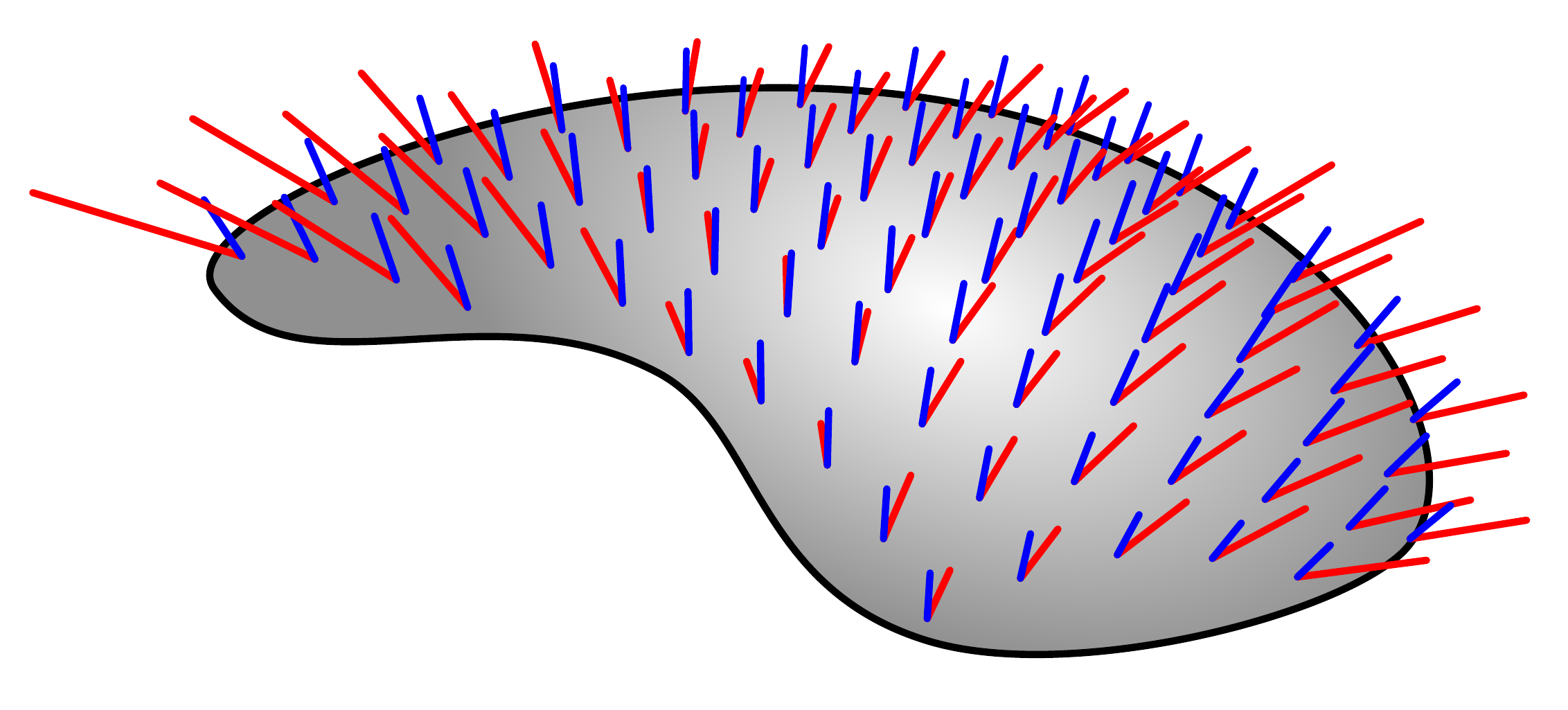}
\caption{Flux density $\mathbf{v}(t,\cdot)$ and unit normal field $\nu$ on surface $S$.}
\end{figure}

Maxwell's tube of fluid motion $T \subset \Omega$ as introduced in property (3) of the Historical Remark \ref{rem:tube} has a bounding surface $S = \partial T$ consisting of a lateral area and two base areas. The flux through the lateral area is zero, since it consists of lines of fluid motion and hence the total flux through the surface $S$ of $T$ is the sum of the fluxes through the two base areas. Maxwell introduces in property (7) of the Historical Remark \ref{rem:tube} sources and sinks, creating and swallowing up the imaginary fluid, respectively, on the base areas of the tube of fluid motion. Maxwell's imaginary fluid has therefore, in particular, a \emph{source density}
\begin{equation*}
   \varrho : \R \times \Omega \to \R,
   \quad
   (t,\mathbf{x}) \mapsto \varrho(t,\mathbf{x}),
\end{equation*}
which specifies by its integral over the tube of fluid motion
\begin{equation*}
   \int_T \varrho (t,\cdot) \,dV
\end{equation*}
the sources and sinks in $T$ which create or swallow up fluid at time $t$. Maxwell's imaginary fluid is incompressible and at any time $t$ the fluid produced by the sources and sinks in $T$ has to equal the flux of the fluid through $S$, i.e.\ the balance law
\begin{equation}\label{eq:balance}
   \int_T \varrho (t,\cdot) \,dV
   =
   \int_{S} \mathbf{v}(t,\cdot) \cdot d \mathbf{S}
\end{equation}
holds. Maxwell uses in \cite{Maxwell1858} the mental imagery of a tube of fluid motion to describe and develop properties of the imaginary fluid. This and various other aspects of the development of electrodynamic theory are discussed in detail in many good papers and books, see e.g.\ \cite{Siegel1991} and the references therein. We leave now Maxwell's original considerations from Remark \ref{rem:tube} but continue in this section to describe his imaginary fluid as an analogy for the development of Maxwell's equations in the next section.

We start with a further simplification of notation by omitting the $t$-dependence of vector fields, i.e.\ we write e.g.\ $\int \mathbf{v} \cdot d \mathbf{S}$ instead of $\int \mathbf{v}(t,\cdot) \cdot d \mathbf{S}$. It is understood that integrands and therefore integral expressions with abbreviated notation do depend on time although $t$ is not explicitly written. Moreover, if an operator $\operatorname{div}$, $\operatorname{curl}$ or $\operatorname{grad}$ is applied to a function which is defined on $\mathbb{R} \times \Omega$, then the partial derivatives involved are only with respect to the space variables in $\Omega$, e.g.\ $\operatorname{div} \mathbf{v}$ is short for $(t,\mathbf{x}) \mapsto \operatorname{div} \mathbf{v}(t,\cdot)|_{\mathbf{x}}$.

The considerations which lead to the balance law \eqref{eq:balance} of Maxwell's imaginary incompressible fluid suggest that it should not only hold for integration domains which are tubes of fluid motion but for each bounded open subset $M \subset \Omega$ with piecewise smooth boundary $\partial M$
\begin{equation}\label{eq:balance-law}
   \int_M \varrho \,dV
   =
   \int_{\partial M} \mathbf{v} \cdot d \mathbf{S}
\end{equation}
ensuring that the fluid produced in $M$ equals the flux through $\partial M$. By the Divergence Theorem \ref{thm:integration}(c)
\begin{equation*}
   \int_M \operatorname{div} \mathbf{v} \,dV
   =
   \int_{\partial M} \mathbf{v} \cdot d\mathbf S
\end{equation*}
and hence the integral $\int_M (\varrho - \operatorname{div} \mathbf{v}) \,dV$ vanishes for each bounded open subset $M \subset \Omega$ with piecewise smooth boundary. As a consequence of Lemma \ref{lem:det_by_int}(c)
\begin{equation*}
   \operatorname{div} \mathbf{v}
   =
   \varrho.
\end{equation*}
So far we have established for Maxwell's imaginary fluid (which has the unit $\textrm{m}^3$ of volume) the properties \emph{(fluid) flux} described by
\vspace*{1ex}
\begin{tcolorbox}[ams align*]
   \begin{aligned}
      \mathbf{v} : \R \times \Omega &\to \R^3,
   \\
      (t,x) &\mapsto \mathbf{v}(t,x)
   \end{aligned}
   \qquad
   \quad
   \text{(\emph{flux density} with unit $[\mathbf{v}] = \tfrac{\textrm{m}}{\textrm{s}}$)}
\end{tcolorbox}
\noindent
and \emph{sinks and sources} described by
\vspace*{1ex}
\begin{tcolorbox}[ams align*]
   \begin{aligned}
      \varrho : \R \times \Omega &\to \R,
   \\
      (t,x) &\mapsto \varrho(t,x)
   \end{aligned}
   \qquad\quad
   \text{(\emph{source density} with unit $[\varrho] = \tfrac{\textrm{m}^3}{\textrm{m}^3 \textrm{s}} = \tfrac{1}{\textrm{s}}$)}
\end{tcolorbox}
\noindent
which are linked by the balance law \eqref{eq:balance-law}, also called \emph{Continuity Equation for the Fluid Flux}, that equates for each bounded open $M \subset \Omega$ with piecewise smooth boundary the sources in $M$ with the flux through $\partial M$
\begin{tcolorbox}[ams align*, title=Continuity Equation for the Fluid Flux (integral form)]
   \begin{aligned}
   \int_M \varrho \,dV
   &=
   \int_{\partial M} \mathbf{v} \cdot d \mathbf{S}
   \qquad\quad
   \text{(with unit $\tfrac{\textrm{m}^3}{s}$)}
   \\
   \text{(\emph{sources in $M$}}
   &=
   \text{\emph{flux through $\partial M$}),}
   \end{aligned}
\end{tcolorbox}
\noindent
or equivalently in differential form
\vspace*{1ex}
\begin{tcolorbox}[ams align*, title=Continuity Equation for the Fluid Flux (differential form)]
   \begin{aligned}
   \varrho
   &=
   \operatorname{div} \mathbf{v}
   \qquad\quad
   \text{(with unit $\tfrac{1}{s}$)}
   \\
   \text{(\emph{source density}}
   &=
   \text{\emph{divergence of flux density}).}
   \end{aligned}
\end{tcolorbox}
\noindent
The source density $\varrho$ and hence also the sources $\int_M \varrho \,dV$ in a bounded open $M \subset \Omega$ with piecewise smooth boundary depend on time $t$. The time derivative $\partial_t \int_M \varrho \,dV = \int_{M} \partial_t \varrho \,dV$ of the sources in $M$ describes the production rate of sources in $M$. The geometric interpretation of Poincar\'{e}'s Lemma \ref{lem:geometric-poincare-lemma}(c), applied to $-\partial_t \varrho$, yields a vector field $\mathbf{j} : \R \times \Omega \to \mathbb{R}^3$ with $\int_{M} \partial_t \varrho \,dV = -\int_{\partial M} \mathbf{j} \cdot d \mathbf{S}$. A positive sign of $\int_{M} \partial_t \varrho \,dV$ means that over time there are more sources in $M$, consequently $\int_{\partial M} \mathbf{j} \cdot d \mathbf{S}$ describes a current of sources from inside of $M$ across the boundary $\partial M$. We therefore have established a \emph{current of sources} described by
\vspace*{1ex}
\begin{tcolorbox}[ams align*]
   \begin{aligned}
      \mathbf{j} : \R \times \Omega &\to \R^3,
   \\
      (t,x) &\mapsto \mathbf{j}(t,x)
   \end{aligned}
   \qquad
   \quad
   \text{(\emph{source current density} with unit $[\mathbf{j}] = \tfrac{\textrm{m}}{\textrm{s}^2}$)}
\end{tcolorbox}
\noindent
which is linked to the \emph{source production rate density} $\partial_t \varrho$ by a balance law that we call \emph{Continuity Equation for the Source Current} and that equates for each bounded open $M \subset \Omega$ with piecewise smooth boundary the source production rate in $M$ with the source current through $\partial M$
\vspace*{1ex}
\begin{tcolorbox}[ams align*, title=Continuity Equation for the Source Current (integral form)]
   \begin{aligned}
   \int_M \partial_t \varrho \,dV
   &=
   - \int_{\partial M} \mathbf{j} \cdot d \mathbf{S}
   \qquad\quad
   \text{(with unit $\tfrac{\textrm{m}^3}{s^2}$)}
   \\
   \text{(\emph{source production rate in $M$}}
   &=
   - \text{\emph{source current through $\partial M$}),}
   \end{aligned}
\end{tcolorbox}
\noindent
or equivalently in differential form
\begin{tcolorbox}[ams align*, title=Continuity Equation for the Source Current (differential form)]
   \begin{aligned}
   \partial_t \varrho
   &=
   - \operatorname{div} \mathbf{j}
   \qquad\quad
   \text{(with unit $\tfrac{1}{s^2}$)}
   \\
   \text{(\emph{source production}}
   &=
   - \text{\emph{divergence of}}
   \\[-1ex]
   \text{\emph{rate density}}
   &
   \hspace*{5ex}
   \text{\emph{source current density}).}
   \end{aligned}
\end{tcolorbox}
\noindent
Note that at this stage of our discussion of the properties of Maxwell's imaginary fluid only the divergence of the source current density $\mathbf{j}$ is uniquely determined and $\mathbf{j}$ plus an arbitrary divergence-free vector field would also be a possible source current density of the imaginary fluid.

We aim now for an application of the geometric interpretation of Poincar\'{e}'s Lemma \ref{lem:geometric-poincare-lemma}(b) to $\partial_t \mathbf{v} + \mathbf{j}$, which is possible because it is divergence-free
\begin{equation*}
   \operatorname{div}
   (\partial_t \mathbf{v} + \mathbf{j})
   =
   \partial_t \operatorname{div} \mathbf{v}
   +
   \operatorname{div} \mathbf{j}
   =   
   \partial_t \varrho - \partial_t \varrho
   = 0.
\end{equation*}
The geometric interpretation of Poincar\'{e}'s Lemma \ref{lem:geometric-poincare-lemma}(b) yields therefore a vector field $\mathbf{h} : \R \times \Omega \to \R^3$ with $\partial_t \mathbf{v} + \mathbf{j} = \operatorname{curl} \mathbf{h}$ and the equivalent geometric interpretation that for each bounded piecewise smooth oriented two-dimensional submani\-fold $M$ of $\Omega$ with piecewise smooth boundary $\partial M$ the balance law $\int_M (\partial_t \mathbf{v} + \mathbf{j}) \cdot d\mathbf{S} = \int_{\partial M} \mathbf{h} \cdot d\mathbf{s}$ holds. It is insightful to explore this balance law for the two special cases (i) $\partial_t \mathbf{v} = 0$ and (ii) $\mathbf{j} = 0$, and $M$ being a surface bounded by a closed curve (visualize $M$ as a disk bounded by a circle $\partial M$). In case (i) of a stationary fluid flux density $\mathbf{v}$ the source current $\int_{M} \mathbf{j} \cdot d \mathbf{S}$ through $M$ equals the line integral $\int_{\partial M} \mathbf{h} \cdot d\mathbf{s}$ of the vector field $\mathbf{h}$ along the closed boundary curve $\partial M$. In case (ii) of a vanishing source current it is the time change $\partial_t \int_M \mathbf{v} \cdot d\mathbf{S}$ of the flux of the imaginary fluid through $M$ which equals $\int_{\partial M} \mathbf{h} \cdot d\mathbf{s}$. From that perspective, the sum of the source current and the flux change rate of the imaginary fluid through the surface $M$ are described by the path integral $\int_{\partial M} \mathbf{h} \cdot d\mathbf{s}$ of $\mathbf{h}$ along the boundary curve $\partial M$. 
We call $\mathbf{h}$ the \emph{fluid field} and $\int_{\partial M} \mathbf{h} \cdot d\mathbf{s}$ the \emph{circulation of the fluid field} along $\partial M$, see Figure \ref{fig:circ}.

\begin{figure}[!h]
\vspace*{-3ex}\hspace*{17ex}
\includegraphics*[width=0.5\textwidth]{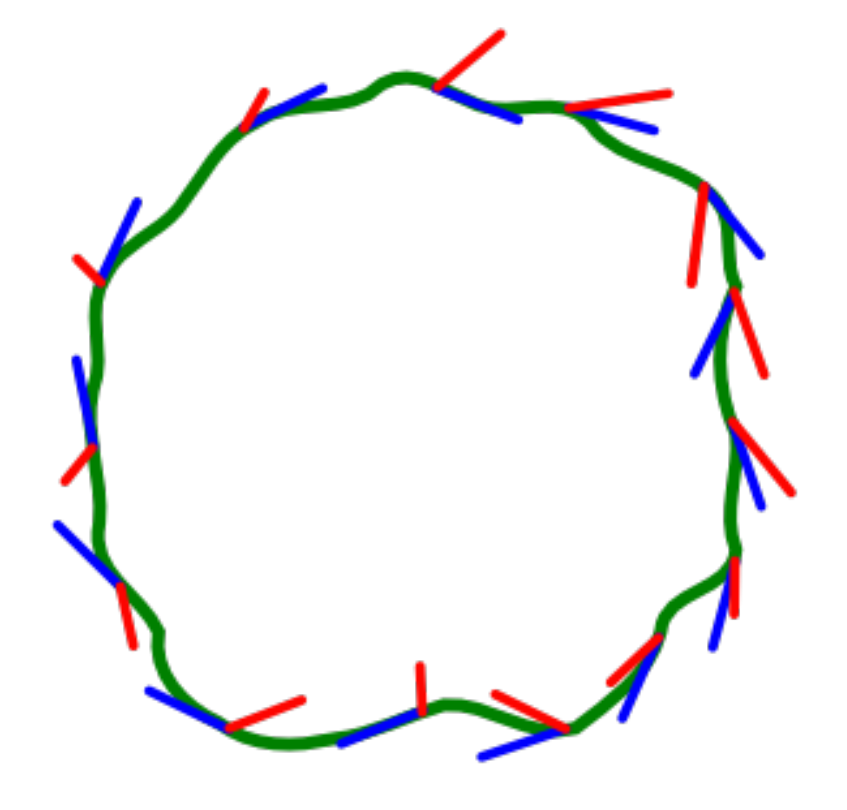}
\vspace*{-2ex}
\caption{Fluid field $\mathbf{h}(t,\cdot)$ and unit tangent field on curve $\partial M$.\label{fig:circ}}
\end{figure}

We have now established
\vspace*{1ex}
\begin{tcolorbox}[ams align*]
   \begin{aligned}
      \mathbf{h} : \R \times \Omega &\to \R^3,
   \\
      (t,x) &\mapsto \mathbf{h}(t,x)
   \end{aligned}
   \qquad
   \quad
   \text{(\emph{fluid field} with unit $[\mathbf{h}] = \tfrac{\textrm{m}^2}{\textrm{s}^2}$)}
\end{tcolorbox}
\noindent
and the change rate $\partial_t \int_M \mathbf{v} \cdot d\mathbf{S}$ of the flux of the imaginary fluid through a surface $M$ as an additive contribution to the source current $\int_{M} \mathbf{j} \cdot d \mathbf{S}$ through $M$. We call $\int_M \partial_t \mathbf{v} \cdot d\mathbf{S}$ the \emph{flux change current} and also for later reference when the imaginary fluid analogy is applied to electricity and magnetism \emph{displacement current}. It is linked to the fluid field $\mathbf{h}$ by a balance law, that we call \emph{Circulation Law for the Fluid Field}, or short \emph{Circulation Law}, and that equates for each bounded piecewise smooth oriented two-dimensional submanifold $M$ of $\Omega$ with piecewise smooth boundary $\partial M$ the sum of the flux change current and the source current through $M$ with the circulation of the fluid field along $\partial M$
\vspace*{1ex}
\begin{tcolorbox}[ams align*, title=Circulation Law (integral form)]
   \begin{aligned}
   \int_M (\partial_t \mathbf{v} + \mathbf{j}) \cdot d\mathbf{S}
   &=
   \int_{\partial M} \mathbf{h} \cdot d\mathbf{s}
   \qquad\quad
   \text{(with unit $\tfrac{\textrm{m}^3}{\textrm{s}^2}$)}
   \\
   \text{(\emph{flux change current $+$\hspace*{3.6ex}}}
   &\phantom{=}
   \\[-1ex]
   \text{\emph{source current through $M$}}
   &=
   \text{\emph{circulation of fluid field along $\partial M$}),}
   \end{aligned}
\end{tcolorbox}
\noindent
or equivalently in differential form
\vspace*{1ex}
\begin{tcolorbox}[ams align*, title=Circulation Law (differential form)]
   \begin{aligned}
   \partial_t \mathbf{v} + \mathbf{j}
   &=
   \operatorname{curl} \mathbf{h}
   \qquad\quad
   \text{(with unit $\tfrac{\textrm{m}}{\textrm{s}^2}$)}
   \\
   \text{(\emph{flux change current $+$\hspace*{1.6ex}}}
   &\phantom{=}
   \\[-1ex]
   \text{\emph{source current densities}}
   &=
   \text{\emph{curl of fluid field}).}
   \end{aligned}
\end{tcolorbox}
\noindent
Note that only the curl of the fluid field $\mathbf{h}$ is uniquely determined and $\mathbf{h}$ plus an arbitrary vector field with vanishing curl would also be a possible fluid field of the imaginary fluid.

\begin{remark}[Interpretation of Circulation Law for Tube of Fluid Motion]\label{rem:cirulationlaw}
For a bounded piecewise smooth oriented two-dimensional submanifold $M$ of $\Omega$ with piecewise smooth boundary $\partial M$ both terms $\int_M \partial_t \mathbf{v} \cdot d\mathbf{S}$ and $\int_M \mathbf{j} \cdot d\mathbf{S}$ with positive sign contribute to an increase of fluid ``on the other side of $M$'', an increase on the side to which the unit normal field of $M$ is pointing. To state this more precisely, consider two cases (i) $M$ is the closed surface of a tube $T$ of fluid motion with outer unit normal field and empty boundary $\partial M$, and case (ii) where $M$ is a base area of a tube $T$ of fluid motion with induced orientation, cp.\ also Figure \ref{fig:fluidmotion}. In case (i), using the fact that an integral over the empty set $\partial M$ equals zero, the Circulation Law becomes $\partial_t \int_M \mathbf{v} \cdot d\mathbf{S} = - \int_M \mathbf{j} \cdot d\mathbf{S}$ with the interpretation that an instantaneous increase of the flux out of $T$ is balanced by a current of sources into $T$. In case (ii) a positive source current $\int_M \mathbf{j} \cdot d\mathbf{S}$ through $M$ out of $T$ contributes to a decrease of fluid production in $T$ by a decrease of sources in $T$. A positive term $\int_M \partial_t \mathbf{v} \cdot d\mathbf{S} = \partial_t \int_M \mathbf{v} \cdot d\mathbf{S}$ can be interpreted as an increase of the flux from inside of $T$ through $M$ out of $T$, because $M$ has the induced orientation of $T$ with an outer unit normal field. The circulation $\int_{\partial M} \mathbf{h} \cdot d\mathbf{s}$ of the fluid field $\mathbf{h}$ along $\partial M$ equals the sum of both terms.
\hfill $\diamond$
\end{remark}

\section{Electricity and Magnetism}

\begin{historicalremark}[Maxwell's Imaginary Fluid Analogy applied to Electricity and Magnetism]\label{rem:electromagnetism}
In \cite{Maxwell1858} Maxwell prepared the application of the imaginary fluid analogy to electricity and magnetism. We quote from \cite[pp.\ 175, 178]{MaxwellNiven1890}.

\begin{quote}\footnotesize
\centerline{\emph{Application of the Idea of Lines of Force.}}
\vspace*{0.2ex}
I have now to shew how the idea of lines of fluid motion as described above may be modified so as to be applicable to the sciences of statical electricity, permanent magnetism, magnetism of induction, and uniform galvanic currents, reserving the laws of electro-magnetism for special consideration. $[\dots]$
\end{quote}

\begin{quote}\footnotesize
Now we found in (18) that the velocity of our imaginary fluid due to a source $S$ at a distance $r$ varies inversely as $r^2$. Let us see what will be the effect of substituting such a source for every particle of positive electricity. $[\dots]$
\end{quote}

\begin{quote}\footnotesize
\centerline{\emph{Theory of Permanent Magnets.}}
\vspace*{1ex}
A magnet is conceived to be made up of elementary magnetized particles, each of which has its own north and south poles, the action of which upon other north and south poles is governed by laws mathematically identical with those of electricity. Hence the same application of the idea of lines of force can be made to this subject, and the same analogy of fluid motion can be employed to illustrate it.
\end{quote}
\end{historicalremark}

We follow Maxwell and apply the analogy of the imaginary fluid to electricity and magnetism. The experimentally observed absence of magnetic monopoles is reflected by setting the \emph{magnetic charge density} $\varrho^{\mathrm{m}}$ and the \emph{magnetic current density} $\mathbf{j}^{\textrm{m}}$ equal to zero. The following table shows the correspondence between concepts for the imaginary fluid and their counterparts for electricity and magnetism.
\medskip

\begin{tabular}{ccc} 
\\[-1ex]
\toprule
   \multicolumn{1}{c}{{Imaginary Fluid}} & \multicolumn{2}{c}{{Electric and Magnetic Concepts}}
   \\ 
   \cmidrule(r){1-1}  \cmidrule(l){2-3}
   \\[-1.7ex]
   \emph{fluid source density} & \emph{electric charge density} & \emph{magnetic charge density}
   \\ \cmidrule(r){1-1}  \cmidrule(lr){2-2} \cmidrule(l){3-3}
\\[-1.9ex]
   $\begin{aligned} 
   \varrho : \R \times \Omega &\to \R,
   \\
   (t,x) &\mapsto \varrho(t,x)
   \end{aligned}$
   &
   $\begin{aligned} 
   \varrho : \R \times \Omega &\to \R,
   \\
   (t,x) &\mapsto \varrho(t,x)
   \end{aligned}$
   &
   $\begin{aligned} 
   \varrho^{\mathrm{m}} : \R \times \Omega &\to \R,
   \\
   (t,x) &\mapsto 0
   \end{aligned}$
\\[2ex]
   with unit $\tfrac{1}{\textrm{s}}$ &
   with unit $\tfrac{\textrm{C}}{\textrm{m}^3}$ &
   with unit $\tfrac{\textrm{T}}{\textrm{m}}$
\\
\midrule
   \\[-1.7ex]
   \emph{fluid flux density} & \emph{electric flux density} & \emph{magnetic flux density}
   \\ \cmidrule(r){1-1}  \cmidrule(lr){2-2} \cmidrule(l){3-3}
\\[-1.9ex]
   $\begin{aligned} 
   \mathbf{v} : \R \times \Omega &\to \R^3,
   \\
   (t,x) &\mapsto \mathbf{v}(t,x)
   \end{aligned}$
   &
   $\begin{aligned} 
   \mathbf{D} : \R \times \Omega &\to \R^3,
   \\
   (t,x) &\mapsto \mathbf{D}(t,x)
   \end{aligned}$
   &
   $\begin{aligned} 
   \mathbf{B} : \R \times \Omega &\to \R^3,
   \\
   (t,x) &\mapsto \mathbf{B}(t,x)
   \end{aligned}$
\\[2ex]
   with unit $\tfrac{\textrm{m}}{\textrm{s}}$ &
   with unit $\tfrac{\textrm{C}}{\textrm{m}^2}$ &
   with unit $\textrm{T}$
\\
\midrule
   \\[-1.7ex]
   \emph{source current density} & \emph{electric current density} & \emph{magnetic current density}
   \\ \cmidrule(r){1-1}  \cmidrule(lr){2-2} \cmidrule(l){3-3}
\\[-1.9ex]
   $\begin{aligned} 
      \mathbf{j} : \R \times \Omega &\to \R^3,
   \\
      (t,x) &\mapsto \mathbf{j}(t,x)
   \end{aligned}$
   &
   $\begin{aligned} 
      \mathbf{j} : \R \times \Omega &\to \R^3,
   \\
      (t,x) &\mapsto \mathbf{j}(t,x)
   \end{aligned}$
   &
   $\begin{aligned} 
      \mathbf{j}^{\textrm{m}} : \R \times \Omega &\to \R^3,
   \\
      (t,x) &\mapsto 0
   \end{aligned}$
\\[2ex]
   with unit $\tfrac{\textrm{m}}{\textrm{s}^2}$ &
   with unit $\tfrac{\textrm{A}}{\textrm{m}^2}$ &
   with unit $\tfrac{\textrm{T}}{\textrm{s}}$
\\
\midrule
   \\[-1.7ex]
   \emph{fluid field} & \emph{magnetic field} & \emph{electric field}
   \\ \cmidrule(r){1-1}  \cmidrule(lr){2-2} \cmidrule(l){3-3}
\\[-1.9ex]
   $\begin{aligned} 
      \mathbf{h} : \R \times \Omega &\to \R^3,
   \\
      (t,x) &\mapsto \mathbf{H}(t,x)
   \end{aligned}$
   &
   $\begin{aligned} 
      \mathbf{H} : \R \times \Omega &\to \R^3,
   \\
      (t,x) &\mapsto \mathbf{H}(t,x)
   \end{aligned}$
   &
   $\begin{aligned} 
      \mathbf{E} : \R \times \Omega &\to \R^3,
   \\
      (t,x) &\mapsto \mathbf{E}(t,x)
   \end{aligned}$
\\[2ex]
   with unit $\tfrac{\textrm{m}^2}{\textrm{s}^2}$ &
   with unit $\tfrac{\textrm{A}}{\textrm{m}}$ &
   with unit $\tfrac{\textrm{V}}{\textrm{m}}$
\\
\bottomrule
\end{tabular}
\medskip

Note that by convention the electric field $\mathbf{E}$ does not correspond to the fluid field $\mathbf{h}$ but to its negative $-\mathbf{h}$, as can be seen below in the analog of the Circulation Law for the Fluid Field (Faraday's Law). The names for $\mathbf{H}$, $\mathbf{D}$, $\mathbf{B}$ and $\mathbf{E}$ emphasize the correspondence to the imaginary fluid analogy. Whereas $\mathbf{E}$ is consistently named \emph{electric field} in the literature, also other names are used for $\mathbf{H}$, $\mathbf{D}$ and $\mathbf{B}$ as listed in the following table, see also \cite{FumeronEtAl2020}.

\medskip
\begin{tabular}{ll} 
   \toprule
   $\mathbf{H}$ & Alternative denominations for \emph{\textbf{magnetic field}}
   \\
   \cmidrule(l){2-2}
   \\[-2.3ex]
   & \emph{magnetic field intensity}, \emph{magnetic field strength}, \emph{magnetizing force}
   \\
   \midrule
   \\[-1.9ex]
   $\mathbf{B}$ & Alternative denominations for \emph{\textbf{magnetic flux density}}
   \\
   \cmidrule(l){2-2}
   \\[-2.3ex]
   & \emph{magnetic induction field}, \emph{magnetic field}
   \\   \midrule
   \\[-1.9ex]
   $\mathbf{D}$ & Alternative denominations for \emph{\textbf{electric flux density}}
   \\
   \cmidrule(l){2-2}
   \\[-2.3ex]
   & \emph{electric displacement field}, \emph{electric induction}
   \\
\bottomrule
\end{tabular}
\medskip

The Continuity Equation for the Fluid Flux of the imaginary fluid now becomes \emph{Gauss' Law} for the electric flux and the \emph{Magnetic Flux Continuity}, respectively. In integral form they hold for each bounded open $M \subset \Omega$ with piecewise smooth boundary $\partial M$.
\vspace*{1ex}
\begin{tcolorbox}[ams align*, title=Gauss' Law (integral form)]
   \begin{aligned}
   \int_M \varrho \,dV
   &=
   \int_{\partial M} \mathbf{D} \cdot d \mathbf{S}
   \qquad\quad
   \text{(with unit $\textrm{C}$)}
   \\
   \text{(\emph{electric charge in $M$}}
   &=
   \text{\emph{electric flux through $\partial M$}),}
   \end{aligned}
\end{tcolorbox}
\noindent
\begin{tcolorbox}[ams align*, title=Magnetic Flux Continuity (integral form)]
   \begin{aligned}
   0
   &=
   \int_{\partial M} \mathbf{B} \cdot d \mathbf{S}
   \qquad\quad
   \text{(with unit $\textrm{T}\textrm{m}^2$)}
   \\
   \text{(\emph{magnetic charge in $M$}}
   &=
   \text{\emph{magnetic flux through $\partial M$}),}
   \end{aligned}
\end{tcolorbox}
\noindent
or equivalently in differential form
\vspace*{1ex}
\begin{tcolorbox}[ams align*, title=Gauss' Law (differential form)]
   \begin{aligned}
   \varrho
   &=
   \operatorname{div} \mathbf{D}
   \qquad\quad
   \text{(with unit $\tfrac{\textrm{C}}{\textrm{m}^3}$)}
   \\
   \text{(\emph{electrical charge density}}
   &=
   \text{\emph{divergence of electric flux density}),}
   \end{aligned}
\end{tcolorbox}
\noindent
\begin{tcolorbox}[ams align*, title=Magnetic Flux Continuity (differential form)]
   \begin{aligned}
   0
   &=
   \operatorname{div} \mathbf{B}
   \qquad\quad
   \text{(with unit $\tfrac{\textrm{T}}{\textrm{m}}$)}
   \\
   \text{(\emph{magnetic charge density}}
   &=
   \text{\emph{divergence of magnetic flux density}).}
   \end{aligned}
\end{tcolorbox}
\noindent
The Continuity Equation for the Source Current of the imaginary fluid becomes the \emph{Continuity Equation for the Electric Current}. The magnetic current is assumed to vanish. In integral form it holds for each bounded open $M \subset \Omega$ with piecewise smooth boundary $\partial M$.
\vspace*{1ex}
\begin{tcolorbox}[ams align*, title=Continuity Equation for the Electric Current (integral form)]
   \begin{aligned}
   \int_M \partial_t \varrho \,dV
   &=
   - \int_{\partial M} \mathbf{j} \cdot d \mathbf{S}
   \qquad\quad
   \text{(with unit $\textrm{A}$)}
   \\
   \text{(\emph{charge production rate in $M$}}
   &=
   - \text{\emph{electric current through $\partial M$}),}
   \end{aligned}
\end{tcolorbox}
\noindent
or equivalently in differential form
\vspace*{1ex}
\begin{tcolorbox}[ams align*, title=Continuity Equation for the Electric Current (differential form)]
   \begin{aligned}
   \partial_t \varrho
   &=
   - \operatorname{div} \mathbf{j}
   \qquad\quad
   \text{(with unit $\tfrac{\textrm{A}}{\textrm{m}^3}$)}
   \\
   \text{(\emph{electric charge}\hspace*{8.5ex}}
   &\phantom{=}\hspace*{1.5ex}
   - \text{\emph{divergence of}}
   \\[-0.8ex]
   \text{\emph{production rate density}}
   &=
   \text{\emph{electric current density}).}
   \end{aligned}
\end{tcolorbox}
\noindent
The change rate of the electric flux $\partial_t \int_M \mathbf{D} \cdot d\mathbf{S}$ through a surface $M$ is called \emph{displacement current} through $M$. The Circulation Law for the Fluid Field applied to the magnetic field $\mathbf{H}$ becomes \emph{Ampère's Law}, also called \emph{Ampère's Circuital Law with Maxwell's Correction}, and it equates for each bounded piecewise smooth oriented two-dimensional submanifold $M$ of $\Omega$ with piecewise smooth boundary $\partial M$ the sum of the displacement current and the electric current through $M$ with the circulation of the magnetic field along $\partial M$.
\vspace*{1ex}
\begin{tcolorbox}[ams align*, title=Ampère's Law (integral form)]
   \begin{aligned}
   \int_M (\partial_t \mathbf{D} + \mathbf{j}) \cdot d\mathbf{S}
   &=
   \int_{\partial M} \mathbf{H} \cdot d\mathbf{s}
   \qquad\quad
   \text{(with unit $\textrm{A}$)}
   \\
   \text{(\emph{displacement current $+$\hspace*{3.6ex}}}
   &\phantom{=}
   \\[-1ex]
   \text{\emph{electric current through $M$}}
   &=
   \text{\emph{circulation of magnetic field along $\partial M$}),}
   \end{aligned}
\end{tcolorbox}
\noindent
or equivalently in differential form
\vspace*{1ex}
\begin{tcolorbox}[ams align*, title=Ampère's Law (differential form)]
   \begin{aligned}
   \partial_t \mathbf{D} + \mathbf{j}
   &=
   \operatorname{curl} \mathbf{H}
   \qquad\quad
   \text{(with unit $\tfrac{\textrm{A}}{\textrm{m}^2}$)}
   \\
   \text{(\emph{displacement current $+$\hspace*{1.6ex}}}
   &\phantom{=}
   \\[-1ex]
   \text{\emph{electric current densities}}
   &=
   \text{\emph{curl of magnetic field}).}
   \end{aligned}
\end{tcolorbox}
\noindent
\begin{remark}[Interpretation of Displacement Current]
As described in Remark \ref{rem:cirulationlaw}, the displacement current $\int_M \partial_t \mathbf{D} \cdot d\mathbf{S}$, as well as the electric current $\int_M \mathbf{j} \cdot d\mathbf{S}$, contribute to a charge increase on the ``other side of $M$''. Whereas the electric current $\int_M \mathbf{j} \cdot d\mathbf{S}$ describes the transport of charge, the displacement current $\partial_t \int_M \mathbf{D} \cdot d\mathbf{S}$ describes an increased flux of electricity. The sum of both yield an electricity current with unit Ampère that equals the circulation $\int_{\partial M} \mathbf{H} \cdot d\mathbf{s}$ of the magnetic field along $\partial M$. For an interpretation of the displacement current which is similar in spirit and also uses the integral representation of Ampère's Law, see \cite{Gauthier1983}.
\hfill $\diamond$
\end{remark}

The Circulation Law for the Fluid Field applied to the electric field $\mathbf{E}$ (which corresponds to $-\mathbf{h}$) becomes \emph{Faraday's Law}, also called \emph{Maxwell-Faraday Equation}. In integral form it holds for each bounded piecewise smooth oriented two-dimen\-sional submanifold $M$ of $\Omega$ with piecewise smooth boundary $\partial M$.
\vspace*{1ex}
\begin{tcolorbox}[ams align*, title=Faraday's Law (integral form)]
   \begin{aligned}
   - \int_M \partial_t \mathbf{B} \cdot d\mathbf{S}
   &=
   \int_{\partial M} \mathbf{E} \cdot d\mathbf{s}
   \qquad\quad
   \text{(with unit $\textrm{V}$)}
   \\
   \text{(\emph{negative of rate of change\hspace*{3.6ex}}}
   &\phantom{=}
   \\[-1ex]
   \text{\emph{of magnetic flux through $M$}}
   &=
   \text{\emph{circulation of electric field along $\partial M$}),}
   \end{aligned}
\end{tcolorbox}
\noindent
or equivalently in differential form
\vspace*{1ex}
\begin{tcolorbox}[ams align*, title=Faraday's Law (differential form)]
   \begin{aligned}
   -\partial_t \mathbf{B}
   &=
   \operatorname{curl} \mathbf{E}
   \qquad\quad
   \text{(with unit $\tfrac{\textrm{V}}{\textrm{m}^2}$)}
   \\
   \text{(\emph{negative of rate of change\hspace*{1.6ex}}}
   &\phantom{=}
   \\[-1ex]
   \text{\emph{of magnetic flux density}}
   &=
   \text{\emph{curl of electric field}).}
   \end{aligned}
\end{tcolorbox}
\noindent
We thus have arrived at Maxwell's equations
\begin{gather*}
   \operatorname{curl} \mathbf{E} + \frac{\partial \mathbf{B}}{\partial t} = 0,
   \quad
   \operatorname{div} \mathbf{D} = \varrho,
\\
   \operatorname{curl} \mathbf{H} - \frac{\partial \mathbf{D}}{\partial t} = \mathbf{j},
   \quad
   \operatorname{div} \mathbf{B} = 0,
\end{gather*}
which are completed by material dependent constitutive relations, e.g.
\begin{equation*}   
   \mathbf{D} = \varepsilon_0 \mathbf{E},
   \quad
   \mathbf{B} = \mu_0 \mathbf{H},
\end{equation*}
in free space with permittivity $\varepsilon_0$ and permeability $\mu_0$. For theoretical and experimental approaches to derive constitutive relations in complex media see e.g.\ \cite{Vinogradov2002} and the references therein.


\end{document}